%% file: paper.tex
\documentclass[11pt]{article}
\usepackage{amsmath,amssymb,amsthm,mathtools}
\usepackage{thm-restate}
\usepackage[margin=1in]{geometry}
\usepackage{moreverb}
\usepackage{hyperref}
\hypersetup{
  colorlinks   = true, 
  urlcolor     = blue, 
  linkcolor    = blue, 
  citecolor   = red 
}
\usepackage{cleveref}
\title{Ramanujan Graphs in Polynomial Time}
\author{
  Michael B. Cohen\\
  M.I.T.\\
  \texttt{micohen@mit.edu}
}

\newtheorem{theorem}{Theorem}[section]
\newtheorem{lemma}{Lemma}[section]

\theoremstyle{definition}
\newtheorem{definition}{Definition}

\begin{document}
\maketitle

\begin{abstract}
The recent work \cite{mss4} by Marcus, Spielman and Srivastava proves the existence of bipartite Ramanujan (multi)graphs of all degrees and all sizes.  However, that paper did not provide a polynomial time algorithm to actually compute such graphs.  Here, we provide a polynomial time algorithm to compute certain expected characteristic polynomials related to this construction.  This leads to a deterministic polynomial time algorithm to compute bipartite Ramanujan (multi)graphs of all degrees and all sizes.

Note: after writing this paper, the author became aware that others have independently produced this result, but does not know any specifics.
\end{abstract}

\input{intro}

\input{interlacing}

\input{simplify}

\input{compute}

\section{Acknowledgements} The author thanks Jonathan Kelner, Daniel Spielman, and Nikhil Srivastava for helpful discussions.  This work was supported by National Science Foundation award 1111109.

\bibliographystyle{alpha}
\bibliography{paper}

\appendix

\input{quadrature}

\end{document}

%% file: intro.tex
\section{Introduction}
Bipartite Ramanujan graphs can be defined as undirected bipartite graphs with constant degree $d$ such that all eigenvalues of the adjacency matrix, except the trivial $\pm d$, have absolute value at most $2 \sqrt{d-1}$.  It is known that this is the smallest bound for which an infinite family--that is, a set of such graphs of degree $d$ containing graphs of arbitrarily large size--can exist (originally due to Alon and Boppana).  Explicit algebraic constructions, naturally providing polynomial time algorithms, of such graphs of certain particular degrees and sizes have been known since the 1980s.  However, until recently it was not known whether infinite families of bipartite Ramanujan graphs existed for all degrees.

Recent papers \cite{mss1,mss4,coverings}, however, have proved the existence of these graphs for all degrees using a powerful new tool known as ``interlacing families.''  They proceed by bounding the roots of the \emph{expected characteristic polynomials} of the adjacency matrices of certain random graphs, then showing that at least one specific graph must satisfy the same bound.  This interlacing family method naturally provides an algorithm to find such a graph.  Unfortunately, the steps of this natural algorithm involve computing certain partially specified expected characteristic polynomials.  There were no known polynomial time algorithms to compute these, so none of these papers provided a polynomial time algorithm to compute Ramanujan graphs.

In this paper, for \cite{mss4} specifically, we provide a polynomial time algorithm to explicitly compute the needed polynomials by reducing to the computation of a certain symbolic determinant.  This provides a polynomial time algorithm to compute the graphs from that paper.  It inherits the mild caveat of \cite{mss4} that the resulting graph may have repeated edges, producing what could be more properly called Ramanujan \emph{multigraphs}.

Additionally, it is still subject to the limitations of the interlacing polynomial method, which does not seem to easily allow producing \emph{non-bipartite} Ramanujan graphs.  Furthermore, this new algorithm does not seem likely to be helpful for \cite{mss1} and some other uses of interlacing families, such as the solution of the Kadison-Singer problem \cite{mss2}.  This is because it is an explicit algorithm for computing the polynomials in the interlacing family, which is known, in the cases of those results, to be \#P-hard.

%% file: interlacing.tex
\section{Interlacing Families}
The construction of \cite{mss4} and the construction of this paper are heavily based on the notion of an \emph{interlacing family}.  We will use a slightly more general definition than is used in that paper, and the specific interlacing family will also be slightly different.

\begin{definition}
An \emph{interlacing family} is a rooted tree with a polynomial $p_k(x)$ assigned to each node $k$, such that:
\begin{enumerate}
\item
All of the polynomials are monic, real-rooted and of the same degree.
\item
For any non-leaf node $k$, $p_k$ is a positive linear combination of the polynomials assigned to the children of $k$.
\item
If two nodes $k$ and $k'$ have the same parent, $p_k$ and $p_{k'}$ have a common interlacing.  That is, the $(j+1)$st root of $p_k$ is not less than the $j$th root of $p_{k'}$ or vice versa.
\end{enumerate}
\end{definition}
For convenience, we will refer to tree nodes in the interlacing family as ``nodes'' and vertices of the graph as ``vertices''.

The relevance of interlacing families stems from the following key lemma (a rephrasing of Lemma 4.2 from \cite{mss1}).
\begin{lemma}
\label{interlacingavg}
Let the real-rooted polynomials $p_i(x)$ all have the same leading term with each pair of polynomials having a common interlacing, and let $p(x)$ be a positive linear combination of the $p_i$.  Then at least one of the $p_i(x)$ must have max root less than or equal to the max root of $p(x)$.
\end{lemma}

This allows one to obtain both existential and computational results from interlacing families:
\begin{lemma}
For any interlacing family with root $r$, there is a leaf node $l$ such that $p_l$ has max root less than or equal to that of $p_r$.
\end{lemma}
\begin{proof}
For any non-leaf node $k$, $k$ must have a child $i$ such that the max root of $p_i$ is less than or equal to the max root of $p_k$, by \Cref{interlacingavg}.  We can then consider the following process: begin with $k = r$.  As long as $k$ is not a leaf, replace $k$ with the child $i$ of $k$ such that $p_i$ has minimum max root.  Once $k$ is a leaf, terminate and return $k$.

The max root of the current node $k$ is monotonically decreasing as this process progresses, so this always returns a leaf node with max root less than or equal to that of $p_r$.
\end{proof}

Now we consider a computational version of this statement.  We define a \emph{computational representation} as a polynomial-time algorithm to exactly compute the coefficients of each polynomial $p_k$ (expressed as rational numbers, with a polynomially bounded number of bits), plus a polynomial time algorithm to list all children of a node.

First, we will need an algorithm to test whether the max root of a polynomial is at most $\sqrt{q}$, for some integer $q$.
\begin{lemma}
\label{testroot}
Let $p$ be a real-rooted polynomial with explicitly given rational coefficients, and $q$ be an integer.  Then one can test whether the max root of $p$ is $\leq \sqrt{q}$ in polynomial time.
\end{lemma}
\begin{proof}
We compute the polynomial $p'(x) = p(x+\sqrt{q})$; we can compute the coefficients of this explicitly in the form $a + b \sqrt{q}$ (i.e. work over $\mathbb{Q}[\sqrt{q}]$), where $a$ and $b$ are rational numbers.  $p$ then has max root $\leq \sqrt{q}$ if and only if $p'$ has max root $\leq 0$ (that is, all of its roots are non-positive).  Since $p'$ is real-rooted (as $p$ is real-rooted), this occurs if and only if all coefficients of $p$ are nonnegative: the ``if'' direction follows from the fact that nonnegative coefficients imply that $p'(x) > 0$ for all $x > 0$, while the ``only if'' direction follows from the factorization $p'(x) = \prod_{\textrm{roots } y} (x-y)$.  Finally, it remains to be able to determine whether $a + b \sqrt{q} \geq 0$ for rational numbers $a$ and $b$.  Here, if $a$ and $b$ have the same sign this is immediate.  If $a$ and $b$ have opposite signs, then $a + b \sqrt{q} \geq 0$ if and only if $a \geq -b \sqrt{q}$.  Squaring both sides, we see that this in turn is true if and only if $a^2 \geq b^2 q$ when $a$ is positive, or $a^2 \leq b^2 q$ if $a$ is negative.
\end{proof}

\begin{lemma}
\label{interlacingalgo}
For any interlacing family root $r$, with polynomially bounded depth and a computational representation, and where $p_r$ has max root at most $\sqrt{q}$ for some integer $q$, then there is a polynomial time algorithm to report a leaf $l$ such that $p_l$ has max root $\leq \sqrt{q}$.
\end{lemma}
\begin{proof}
The proof is essentially the same.  Again, we track a ``current'' node $k$, which is initialized to $r$; we will maintain the invariant that the max root of $p_k$ is at most $\sqrt{q}$.  As long as $k$ is not a leaf, we list its children $i$, and compute $p_i$ for each.  Next, we test whether each of these $p_i$ has max root $\leq \sqrt{q}$ using \Cref{testroot}.  By \Cref{interlacingavg}, at least one of these children must have max root at most $\sqrt{q}$; we then set $k$ to this child, maintaining the invariant.  Once $k$ is a leaf, we terminate and return it.  The total number of iterations is at most the depth of the tree.
\end{proof}

\section{The Matching Interlacing Family}
Our aim is to use an interlacing family to find a bipartite Ramanujan graph of degree $d$ on $n$ vertices.

Here, we are interested in obtaining a graph whose adjacency matrix has bounded max nontrivial eigenvalue--we crucially use the fact that the eigenvalues of the adjacency matrix of a bipartite graph are symmetric around 0, and thus if the maximum nontrivial eigenvalue is at most $2 \sqrt{d-1}$, the minimum is at least $-2 \sqrt{d-1}$ and the graph is Ramanujan.

To get this out of an interlacing family, we would like the leaves of that family to correspond to the characteristic polynomials, divided by the factors from the trivial eigenvalues, of the adjacency matrices of $d$-regular bipartite graphs on $n$ vertices.  We will construct such an interlacing family, which we will call the Matching Interlacing Family.  This family is slightly different from the one presented in \cite{mss4}, since that one has an exponentially deep tree and so is not directly suitable for algorithmic use.

To get this, we will look at the combination of $d$ perfect bipartite matchings.  We index the vertices of the graph so that vertices 1 through $\frac{n}{2}$ are on one side of the bipartition and vertices $\frac{n}{2}+1$ through $n$ are on the other.  Then we define a \emph{partially specified matching} as a bipartite matching that matches vertices 1 through $t$, for some $0 \leq t < \frac{n}{2}$, and does not match any vertices from $t+1$ through $\frac{n}{2}$.

Then we will define a node in our family as corresponding to a sequence of $r$ bipartite matchings, for $r \leq d$, such that all except the last are complete and the last is either a complete bipartite matching or a partially specified matching.  The children of a node are those that add the next unmatched vertex to the partially specified matching, or begin a new partially specified matching (matching vertex 1) if the last matching is complete.  The leaves are sequences of $d$ perfect matchings, which combine to a $d$-regular bipartite graph on $n$ vertices.

For each node in the interlacing family, we furthermore view it as corresponding to a probability distribution over $d$-regular bipartite graphs on $n$ vertices.  This is simply uniformly randomly completing the partially specified matching if present, then adding $d-r$ uniformly random complete bipartite matchings.  It may be immediately seen that the probability distribution assigned to a node is the average of the distributions assigned to its children, and the root is the combination of $d$ random bipartite graphs.

Finally, we define the polynomial assigned to a node as the \emph{expected} characteristic polynomial, over the distribution of graphs associated with the node, of the adjacency matrix, divided by the factors from the trivial eigenvalues $(x+d)(x-d)$.

\begin{lemma}
\label{matchinginterlacing}
The Matching Interlacing Family is an interlacing family.
\end{lemma}

The proof of this lemma depends on several results from \cite{mss4}:
\begin{lemma}[The second part of Lemma 2.2 from \cite{mss4}]
\label{avgrealrooted}
Two monic real-rooted polynomials of the same degree, $p_1$ and $p_2$, have a common interlacing if and only if all convex combinations of $p_1$ and $p_2$ are real-rooted.
\end{lemma}

\begin{theorem}[Theorem 3.3 from \cite{mss4}]
\label{swapsrealrooted}
Let $A_i$ be arbitrary deterministic symmetric $d \times d$ matrices, and let $S_{ij}$ be independent random swap matrices (that is, $S_{ij}$ swaps two indices $s_{ij}$ and $t_{ij}$ with some probability $\alpha_{ij}$, and is the identity with probability $1-\alpha_{ij}$).  Then the expected characteristic polynomial of the random matrix
\begin{equation*}
\sum_i \left ( \prod_{j=1}^{N_i} S_{ij} \right ) A_i \left ( \prod_{j=1}^{N_i} S_{ij} \right )^T
\end{equation*}
is real-rooted.
\end{theorem}

\begin{lemma}[Lemma 3.5 from \cite{mss4}]
\label{permutationbyswaps}
A uniformly random permutation matrix $P$ can be expressed as $\prod_{j=1}^N S_j$, where $S_j$ are independent random swap matrices (in the same sense as in \Cref{swapsrealrooted}).
\end{lemma}

\begin{proof}[Proof of \Cref{matchinginterlacing}]
Since the distribution associated with a non-leaf node is the average of the distribution of its children, it immediately follows that the polynomial associated with a non-leaf node is the average (which is, in particular, a positive linear combination) of the polynomials associated with its children.  All of the characteristic polynomials are of $n \times n$ matrices, so they all are all monic polynomials of the same degree.  It remains to show that the polynomials are real-rooted and that any two siblings: that is, the same sequences of matchings followed by partially specified matchings differing only in the assignment of the last matched vertex, are assigned polynomials with a common interlacing.

To do this, we first note that interlacing and real-rootedness is not affected by dividing out the terms from the trivial eigenvalues, since that always removes the same real roots from all polynomials.  Therefore, we just need to show that the expected characteristic polynomials themselves are real-rooted and interlacing for siblings.

We note that we can write the random adjacency matrix for a node $k$ as
\begin{equation*}
\left ( \sum_{i=1}^{r-1} P^k_i M (P^k_i)^T \right ) + P^k_r M (P^k_r)^T + \sum_{i=r+2}^d P_i M P_i^T
\end{equation*}
where $M$ is an arbitrary bipartite matching, $P^k_i$ for $i < r$ are deterministic permutation matrices that permute it into the already-selected complete matchings, $P^k_r$ is a partially random permutation providing the matched edges and the remaining random edges, and $P_i$ for $i > r$ are uniformly random permutations.  We may apply \Cref{permutationbyswaps} to express $P^k_r$ as a product of independent random swaps times a deterministic matrix and $P_i$ for $i > r$ as a product of independent random swaps.  This puts it in the form needed to apply \Cref{swapsrealrooted}, implying that this expected characteristic polynomial is real-rooted.

To get the interlacing, we note that for any sibling node $k'$, we may express its random adjacency matrix as
\begin{equation*}
\left ( \sum_{i=1}^{r-1} P^k_i M (P^k_i)^T \right ) + S P^k_r M (P^k_r)^T S^T + \sum_{i=r+2}^d P_i M P_i^T
\end{equation*}
where $S$ is a deterministic swap matrix that simply swaps the two different partners of the last matched node in the siblings.  Any convex combination with the sibling then can be expressed as an expected characteristic polynomial of the same form, except with $S$ now a \emph{random} swap matrix with some probability.  We can then apply \Cref{permutationbyswaps} and \Cref{swapsrealrooted} in the same manner as before, showing that any convex combination of the expected characteristic polynomials for two siblings is real-rooted.  Finally, by \Cref{avgrealrooted}, this means the two siblings have a common interlacing.
\end{proof}

Furthermore, the root node of the Matching Interlacing Family is the expected characteristic polynomial of the combination of $d$ uniformly random complete bipartite matchings on $n$ vertices, divided by the trivial factors.  It is proved in \cite{mss4} that this has max root bounded by $2 \sqrt{d-1}$.  Therefore, \Cref{interlacingalgo}, applied to this interlacing family, would produce a bipartite Ramanujan graph.  However, this requires a computational representation.  Using an explicit representation of the matchings to index the nodes, the ability to efficiently list the children is immediate.  What remains is to give a polynomial-time algorithm to compute the needed expected characteristic polynomial (the naive algorithm, simply summing over all possible assignments, is far from being polynomial time).

%% file: simplify.tex
\section{Simplifying the Problem}
Fortunately, we can substantially simplify the polynomial computation problem here.

First, consider an arbitrary $c$-regular bipartite graph $G$, and let $G'$ be $G$ plus a uniformly random bipartite matching.  Now, we note that Theorem 4.10 of \cite{mss4} (proved via \cite{finiteconv}) gives a formula for the expected characteristic polynomial of the adjacency matrix of of $G'$ in terms of the characteristic polynomial of $G$ itself.  Furthermore, for any fixed degree, this formula is \emph{linear}--it expresses the coefficients of the expected characteristic polynomial of $G'$ as linear combinations of the coefficients of the characteristic polynomial of $G$--and directly allows computation in polynomial time.  By linearity of expectation, even if $G$ itself is a random graph, as long as it has known constant degree the expected characteristic polynomial of $G'$ can be computed from the expected characteristic polynomial of $G$.

This means that to compute the expected characteristic polynomial of the distribution of graphs associated with a node of the Matching Interlacing Family, it suffices to be able to do it in the case where the partially specified matching is the last one.  We can always simply take the expected characteristic polynomial for the graph consisting of the first $r$ matchings only, then linearly transform this to account for the $d-r$ remaining uniformly random matchings.

This amounts to taking the expected characteristic polynomial of a bipartite graph (the complete matchings plus the edges already selected in the partially specified matching) plus a random bipartite matching precisely covering some given \emph{subset} of the vertices (which will be the unmatched vertices from the partially specified matching).  This is the expected characteristic polynomial of a random matrix of the form

\begin{equation*}
\begin{pmatrix}
0 & (A + P_B) \\
(A + P_B)^T & 0
\end{pmatrix}
\end{equation*}

where $A$ is a fixed matrix and $P_B$ is a random permutation on some sub-block $B$ (with rows and columns corresponding to unmatched vertices on both sides), and zero outside the block.

Next, we note that one can perform the quadrature argument from \cite{mss4}.  We perform orthogonal changes of basis to the rows and columns of $A + P_B$ (which always preserves the characteristic polynomial of interest) that keep the block structure of $B$, but transform the all-ones vector within the block to a basis vector.  This isolates the non-mean-zero part of the permutation matrix, as in \cite{mss4}'s quadrature.  We may fold that part into $A$, asking for the expected characteristic polynomial of
\begin{equation*}
\begin{pmatrix}
0 & (\hat A + \hat P_{\hat B}) \\
(\hat A + \hat P_{\hat B})^T & 0
\end{pmatrix}
\end{equation*}
where $\hat A$ is the change of variables applied to $A + \mathrm{E}[P_B]$, and $\hat P_{\hat B}$ is the change of variables applied to $P_B-\mathrm{E}[P_B]$ (noting that $\mathrm{E}[P_B]$ is precisely the part of $P_B$ aligned with the all-ones vector on the block).  We define $\hat B$ as the containing the directions from $B$ orthogonal to the all-ones vector.  $\hat B$ is $(l-1)$-dimensional if $B$ was $l$-dimensional.

Note that $\hat P_{\hat B}$ are not actually permutation matrices in that block, but the result of a change of variables of a permutation matrix (minus its all-ones component) in the original basis.  More specifically, consider the unit basis vectors $e_i$ in the original block $B$.  If we project off their component in the direction of the all-ones vector, then apply the change of variables, we get a regular simplex of $l$ vectors $\hat e_i$, with all pairs at distance exactly $\sqrt{2}$ centered at the origin, in $\mathbb{R}^{l-1}$.  $\hat P_{\hat B}$ is then a linear transformation that randomly permutes the $\hat e_i$ (while acting as the identity outside the block $\hat B$).

Now, we have the main quadrature result:
\begin{restatable}{theorem}{quadrature}
\label{quadrature}
Let $M$ be an arbitrary $n \times n$ symmetric matrix, and let $\hat P_{\hat B}$ be a random matrix that randomly permutes the vertices of an $l-1$-dimensional regular simplex, centered at 0, within an $l-1$-dimensional block $\hat B$, while acting as the identity outside $\hat B$.  Then the expected characteristic polynomial of
\begin{equation*}
M + \begin{pmatrix}
0 & \hat P_{\hat B} \\
\hat P_{\hat B}^T & 0
\end{pmatrix}
\end{equation*}
is equal to the expected characteristic polynomial of
\begin{equation*}
M + \begin{pmatrix}
0 & Q_{\hat B} \\
Q_{\hat B}^T & 0
\end{pmatrix}
\end{equation*}
where $Q_{\hat B}$ is a Haar-random (i.e. uniformly random) orthogonal matrix on the block $\hat B$.
\end{restatable}
This implies that we can instead ask for the expected characteristic polynomial of
\begin{equation*}
\begin{pmatrix}
0 & (\hat A + Q_{\hat B}) \\
(\hat A + Q_{\hat B})^T & 0
\end{pmatrix}
\end{equation*}

Finally, we note that this is a symmetric block-off-diagonal matrix, allowing us to apply the following lemma:
\begin{lemma}
Given a square matrix $M$, let $p$ be the characteristic polynomial of the block-off-diagonal matrix
\begin{equation*}
N = \begin{pmatrix}
0 & M \\
M^T & 0
\end{pmatrix}
\end{equation*}
and $p'$ be the characteristic polynomial of $M^T M$.  Then $p(x) = p'(x^2)$.
\end{lemma}
\begin{proof}
This follows from the fact that the eigenvalues of $N$ can be placed in correspondence to those of $M^T M$: an eigenvalue of $y$ in $M^T M$ corresponds to two eigenvalues $\pm \sqrt{y}$ in $N$.  That in turn means that each factor $(x-y)$ of $p'$ becomes a factor $(x+\sqrt{y})(x-\sqrt{y}) = (x^2-y)$ in $p$.
\end{proof}
That means that our characteristic polynomial could be obtained by taking the characteristic polynomial of $(\hat A + Q_{\hat B})^T (\hat A + Q_{\hat B})$ and replacing $x$ with $x^2$.  Since this is again a linear map on the coefficients, the expected characteristic polynomial of the block off-diagonal matrix can also be obtained by applying this transformation to the expected characteristic polynomial of $(\hat A + Q_{\hat B})^T (\hat A + Q_{\hat B})$.  Computing expected characteristic polynomials for matrices of that form is therefore sufficient.

%% file: compute.tex
\section{Computing the Expected Characteristic Polynomial}
The $(n/2-k)$th coefficient of the characteristic polynomial of an $(n/2) \times (n/2)$ matrix $Z^T Z$ is equal to $(-1)^{n/2-k}$ times the sum of squares of all $k \times k$ minors of $Z$.  For $Z = \hat A + Q_{\hat B}$, we can express this sum as
\begin{equation*}
\sum_{U,V \textrm{ s.t. } |U| = |V| = k} \det((\hat A + Q_{\hat B})_{U,V})^2
\end{equation*}
where $U$ and $V$ are subsets of rows and columns, respectively, and $M_{U,V}$ is the submatrix of $M$ containing the rows in $U$ and the columns in $V$.

By the multilinearity of the determinant, we may write
\begin{equation*}
\det((\hat A + Q_{\hat B})_{U,V}) = \sum_{U' \subseteq U, V' \subseteq V, |U'| = |V'|} \operatorname{sign}(U, V, U', V') \det(\hat A_{U',V'}) \det((Q_{\hat B})_{U \setminus U',V \setminus V'})
\end{equation*}
where $\operatorname{sign}(U,V,U',V')$ is $\pm 1$ depending on the sets.  We may then look at the expected value of the square of this determinant:
\begin{align*}
\hspace{2em}&\hspace{-2em} E_{Q_{\hat B}}[\det((\hat A + Q_{\hat B})_{U,V})^2] \\
= &\sum_{\substack {U' \subseteq U, V' \subseteq V, |U'| = |V'|, \\ U'' \subseteq U, V'' \subseteq V, |U''| = |V''|}} \left ( \begin{multlined}
\operatorname{sign}(U,V,U',V') \operatorname{sign}(U,V,U'',V'') \det(\hat A_{U',V'}) \det(\hat A_{U'',V''}) \\
E_{Q_{\hat B}}[\det((Q_{\hat B})_{U \setminus U',V \setminus V'}) \det((Q_{\hat B})_{U \setminus U'',V \setminus V''})]
\end{multlined} \right )
\end{align*}
Luckily, whenever $U' \neq U''$ or $V' \neq V''$, at least one row or column of $(Q_{\hat B})_{U \setminus U',V \setminus V'}$ or $(Q_{\hat B})_{U \setminus U'',V \setminus V''}$ is a symmetric random variable conditioned on the other one, and so 
\begin{equation*}
E_{Q_{\hat B}}[\det((Q_{\hat B})_{U \setminus U',V \setminus V'}) \det((Q_{\hat B})_{U \setminus U'',V \setminus V''})] = 0
\end{equation*}
We can then write the expected value as
\begin{equation*}
E_{Q_{\hat B}}[\det((\hat A + Q_{\hat B})_{U,V})^2] = \sum_{U' \subseteq U, V' \subseteq V, |U'| = |V'|} E_{Q_{\hat B}}[\det((Q_{\hat B})_{U \setminus U',V \setminus V'})^2] \det(\hat A_{U',V'})^2
\end{equation*}

We define
\begin{equation*}
f_{\hat B}(U,V,U',V') = E_{Q_{\hat B}}[\det((Q_{\hat B})_{U \setminus U',V \setminus V'})^2].
\end{equation*}

If $U \setminus U'$ or $V \setminus V'$ is not contained within the block $\hat B$, $(Q_{\hat B})_{U \setminus U',V \setminus V'}$ includes a zero row or column, so
\begin{equation*}
f_{\hat B}(U,V,U',V) = 0
\end{equation*}
Otherwise,
\begin{equation*}
f_{\hat B}(U,V,U',V') = \frac{1}{\dbinom{\hat l}{|U|-|U'|}}
\end{equation*}
by a symmetry argument (where $\hat l$ is the dimension of the block $\hat B$).

The expected value of the sum of the squares of all of the $k \times k$ minors is then
\begin{equation*}
\sum_{U,V \textrm{ s.t. } |U| = |V| = k} \left ( \sum_{U' \subseteq U, V' \subseteq V, |U'| = |V'|} f_{\hat B}(U, U', V, V') \det(\hat A_{U',V'})^2 \right )
\end{equation*}
Switching the order of summation, we can write it as
\begin{equation*}
\sum_{U',V'\textrm{ s.t. }|U'| = |V'|} \left ( \sum_{U,V \textrm{ s.t. } U' \subseteq U, V' \subseteq V, |U| = |V| = k} f_{\hat B}(U, U', V, V') \right ) \det(\hat A_{U',V'})^2
\end{equation*}

But now we can note that
\begin{equation*}
\sum_{U,V \textrm{ s.t. } U' \subseteq U, V' \subseteq V, |U| = |V| = k} f_{\hat B}(U, U', V, V') = \frac{{\dbinom{\hat l-O_{\hat B}^r(U')}{k-|U'|}} {\dbinom{\hat l-O_{\hat B}^c(V')}{k-|U'|}}}{\dbinom{\hat l}{k-|U'|}}
\end{equation*}
where $\hat l$ is the dimension of the block $\hat B$, $O_{\hat B}^r(U')$ gives the number of rows in $U'$ contained in $\hat B$ and $O_{\hat B}^c(V')$ is analogous with columns.  This follows simply follows from counting the number of choices of $U$ and $V$ for which $f_{\hat B}(U, U', V, V')$ is nonzero (precisely those $U$ and $V$ that make $U \setminus U'$ and $V \setminus V'$ fully contained in the block $\hat B$), combined with the value of $f$ when it is nonzero.

We define
\begin{equation*}
g(\hat l, k, k', p, q) = \frac{{\dbinom{\hat l-p}{k-k'}} {\dbinom{\hat l-q}{k-k'}}}{\dbinom{\hat l}{k-k'}}
\end{equation*}

We may then write the expected sum of the squares of the $k \times k$ minors as
\begin{equation*}
\sum_{k',p,q} g(\hat l, k, k', p, q) \left ( \sum_{\substack{U',V' \textrm{ s.t. } |U'| = |V'| = k', \\ O_{\hat B}^r(U') = p, O_{\hat B}^c(V') = q}} \det(\hat A_{U',V'})^2 \right )
\end{equation*}

That is, we can compute the coefficients of our expected characteristic polynomial as linear functions in the values
\begin{equation*}
C_{k',p,q} = \left ( \sum_{\substack{U',V' \textrm{ s.t. } |U'| = |V'| = k', \\ O_{\hat B}^r(U') = p, O_{\hat B}^c(V') = q}} \det(\hat A_{U',V'})^2 \right )
\end{equation*}

It remains to be able to compute these values.  Note that without the constraints of $O_{\hat B}^r$ and $O_{\hat B}^c$, this would just be a sum of squared minors of a particular size, which correspond to the coefficients of the characteristic polynomial of $\hat A^T \hat A$.  When $\hat B$ covers the entire matrix, they must simply be equal to $k'$, so the expected characteristic polynomial can therefore be computed using only the characteristic polynomial of $\hat A^T \hat A$ itself. That makes sense as that is a special case of a finite free convolution from \cite{finiteconv}.  This does not work when $\hat B$ does not cover the whole matrix.

However, we can introduce a \emph{trivariate} determinant polynomial, extending the standard characteristic polynomial
\begin{equation*}
P(\lambda, t_r, t_c) = \det(\hat A^T ((I-D_r^{\hat B})+t_r D_r^{\hat B}) \hat A ((I-D_c^{\hat B})+t_c D_c^{\hat B}) + \lambda I)
\end{equation*}
where $D_r^{\hat B}$ is the diagonal matrix with ones in the rows in $\hat B$ and zeroes in the others, and $D_c^{\hat B}$ the same for the columns in $\hat B$.  The coefficient of this trivariate polynomial on $\lambda^{(n/2-k')} t_r^p t_c^q$ now provides precisely the value $C_{k',p,q}$ that we are looking for.

One way to see this is to note that that for fixed $t_r$, $t_c$, it is $\det(\bar A^T \bar A + \lambda I)$, where $\bar A$ is $\hat A$ with the rows in $\hat B$ multiplied by $\sqrt{t_r}$ and the columns in $B$ multiplied by $\sqrt{t_c}$.  The $\lambda^{(n/2-k')}$ coefficient of that polynomial is the sum of the squares of the $k' \times k'$ minors of $\bar A$.  The extra multiplications by $\sqrt{t_r}$ and $\sqrt{t_c}$ multiply the contribution of a minor with $p$ rows and $q$ columns overlapping the block by by $t_r^p$ and $t_c^q$.  Thus each such minor has its contribution overall scaled by $\lambda^{(n/2-k')} t_r^p t_c^q$.

Furthermore, as there are only a constant number of variables and the matrix is only constant degree in each, we can simply evaluate the determinant of this symbolically in polynomial time.  This gives us polynomial time computation of the needed expected characteristic polynomials and therefore, by \Cref{interlacingalgo}, the Ramanujan graphs.

%% file: quadrature.tex
\section{Quadrature argument}
\quadrature*
We will work up to the proof of this theorem, using results from \cite{mss4}.

Now we consider a random matrix $\hat P_{ijk}$ which randomly permutes three vectors, $\hat e_i$, $\hat e_j$, and $\hat e_k$, forming a standard simplex, while leaving all vectors orthogonal to $\hat e_j - \hat e_i$ and $\hat e_k - \hat e_i$ unaffected.  Similarly, we look at $Q_{ijk}$ which is a random rotation on the two-dimensional subspace spanned by $\hat e_j - \hat e_i$ and $\hat e_k - \hat e_i$.  \cite{mss4} gives us a very useful result about these matrices:
\begin{lemma}[Corollary 4.6 of \cite{mss4}]
\label{triangleswap}
For any fixed symmetric matrices $A$, $B$, the expected characteristic polynomial of $A + \hat P_{ijk} B \hat P_{ijk}^T$ is equal to that of $A + Q_{ijk} B Q_{ijk}^T$.
\end{lemma}

That in turn implies
\begin{lemma}
\label{eatq}
Let $X$ be the random matrix
\begin{equation*}
\begin{pmatrix}
0 & \hat P_{\hat B} \\
\hat P_{\hat B}^T & 0
\end{pmatrix}.
\end{equation*}
and let $Q_{ijk}^{(0)}$ be any orthogonal matrix acting within the span of $\hat e_j - \hat e_i$ and $\hat e_k - \hat e_i$ (in the upper $n/2 \times n/2$ block of the $n \times n$ matrix).  Then for any fixed matrix $M$, the expected characteristic polynomial (over $X$) of $(Q_{ijk}^{(0)})^T M Q_{ijk}^{(0)} + X$ is equal
to that of $M + X$.
\end{lemma}
\begin{proof}
We first note that the expected characteristic polynomial of $(Q_{ijk}^{(0)})^T M Q_{ijk}^{(0)} + X$ is equal to that of (taking the expectation over $X$ and $P_{ijk}$)
\begin{equation*}
(Q_{ijk}^{(0)})^T M Q_{ijk}^{(0)} + \hat P_{ijk} X \hat P_{ijk}^T
\end{equation*}
since the extra $P_{ijk}$ leaves the probability distribution of $X$ unchanged.  Now, applying \Cref{triangleswap}, this is the same as the expected characteristic polynomial of
\begin{equation*}
(Q_{ijk}^{(0)})^T M Q_{ijk}^{(0)} + Q_{ijk} X Q_{ijk}^T
\end{equation*}
Using invariance of the characteristic polynomial under conjugation, this is the same as the expected characteristic polynomial of
\begin{equation*}
M + Q_{ijk}^{(0)} Q_{ijk} X Q_{ijk}^T (Q_{ijk}^{(0)})^T
\end{equation*}
which has the same distribution as
\begin{equation*}
M + Q_{ijk} X Q_{ijk}^T
\end{equation*}
Finally, applying \Cref{triangleswap} again, we get the expected characteristic polynomial of
\begin{equation*}
M + \hat P_{ijk} X \hat P_{ijk}^T
\end{equation*}
which has the same probability distribution as $M + X$, as desired.
\end{proof}

Now we pull in another result from \cite{mss4}:
\begin{lemma}[Lemma 4.7 of \cite{mss4}]
\label{generates}
Matrices of the form of $Q_{ijk}^{(0)}$--that is, orthogonal transformations affecting only the plane spanned by some three vertices of our simplex $\hat e_i$, $\hat e_j$, and $\hat e_k$--generate the group of all orthogonal transformations within $\hat B$.  In other words, any orthogonal matrix in $\hat B$ can be expressed as a product of a finite number of matrices of this type.
\end{lemma}

Putting this together we get
\begin{lemma}
\label{anyshift}
For any orthogonal matrix $Q^{(0)}$ on the block $\hat B$, the expected characteristic polynomial of
\begin{equation*}
M + Q^{(0)} X (Q^{(0)})^T
\end{equation*}
is equal to the expected characteristic polynomial of $M + X$
\end{lemma}

\begin{proof}
First, use the invariance of the characteristic polynomial under conjugation to express this as the expected characteristic polynomial of
\begin{equation*}
(Q^{(0)})^T M Q^{(0)}
\end{equation*}
Then applying \Cref{generates} expresses $Q^{(0)}$ as a product of matrices of the type of $Q_{ijk}^{(0)}$.  However, \Cref{eatq} says that conjugating $M$ by any of these matrices does not affect the expected characteristic polynomial; thus, iterating it implies that conjugating $M$ by any product of these matrices does not affect it, so the expected characteristic polynomial is that that of $M + X$, as desired.
\end{proof}

Finally, we may prove \Cref{quadrature}:
\begin{proof}[Proof of \Cref{quadrature}]
The probability distribution of
\begin{equation*}
\begin{pmatrix}
0 & Q_{\hat B} \\
Q_{\hat B}^T & 0
\end{pmatrix}
\end{equation*}
is the same as
\begin{equation*}
\begin{pmatrix}
Q_{\hat B} & 0 \\
0 & I
\end{pmatrix} X \begin{pmatrix}
Q_{\hat B} & 0 \\
0 & I
\end{pmatrix}^T
\end{equation*}
since multiplying by $Q_{\hat B}$ will turn any fixed orthogonal matrix into a random orthogonal matrix.  But by \Cref{anyshift}, for any fixed value of $Q_{\hat B}$,
\begin{equation*}
M + \begin{pmatrix}
Q_{\hat B} & 0 \\
0 & I
\end{pmatrix} X \begin{pmatrix}
Q_{\hat B} & 0 \\
0 & I
\end{pmatrix}^T
\end{equation*}
has the same expected characteristic polynomial (randomizing over only $X$) as $M + X$.
Thus this must still be true for random $Q_{\hat B}$, implying that
\begin{equation*}
M + \begin{pmatrix}
0 & Q_{\hat B} \\
Q_{\hat B}^T & 0
\end{pmatrix}
\end{equation*}
has the same expected characteristic polynomial as $M + X$, as desired.
\end{proof}

%% file: paper.bbl
\begin{thebibliography}{MSS15b}

\bibitem[HPS15]{coverings}
Chris Hall, Doron Puder, and William~F Sawin.
\newblock Ramanujan coverings of graphs.
\newblock {\em arXiv preprint arXiv:1506.02335}, 2015.

\bibitem[MSS13a]{mss1}
A.~Marcus, D.~A. Spielman, and N.~Srivastava.
\newblock Interlacing families i: Bipartite ramanujan graphs of all degrees.
\newblock In {\em Foundations of Computer Science (FOCS), 2013 IEEE 54th Annual
  Symposium on}, pages 529--537, Oct 2013.

\bibitem[MSS13b]{mss2}
Adam Marcus, Daniel~A Spielman, and Nikhil Srivastava.
\newblock Interlacing families ii: Mixed characteristic polynomials and the
  kadison-singer problem.
\newblock {\em arXiv preprint arXiv:1306.3969}, 2013.

\bibitem[MSS15a]{mss4}
A.~W. Marcus, D.~A. Spielman, and N.~Srivastava.
\newblock Interlacing families iv: Bipartite ramanujan graphs of all sizes.
\newblock In {\em Foundations of Computer Science (FOCS), 2015 IEEE 56th Annual
  Symposium on}, pages 1358--1377, Oct 2015.

\bibitem[MSS15b]{finiteconv}
Adam Marcus, Daniel~A Spielman, and Nikhil Srivastava.
\newblock Finite free convolutions of polynomials.
\newblock {\em arXiv preprint arXiv:1504.00350}, 2015.

\end{thebibliography}
